\documentclass[runningheads,envcountsame]{llncs}

\usepackage{amsthm}
\usepackage[T1]{fontenc}
\usepackage{makecell}
\usepackage{multirow}
\usepackage{caption}
\usepackage{arydshln}
\usepackage{xstring}
\usepackage{amsmath}
\usepackage{amsfonts} 
\usepackage{amssymb} 
\usepackage[shortlabels]{enumitem}
\usepackage{tikz}
\usetikzlibrary{decorations.markings}
\usepackage{graphicx}
\usepackage[linesnumbered,ruled]{algorithm2e}
\usepackage{lineno}
\usepackage{hyperref}
\usepackage{float}
\usepackage{booktabs}
\usepackage{fp}

\newcommand*{\Scale}[2][4]{\scalebox{#1}{\ensuremath{#2}}}%
\newcommand\myshade{85}
\colorlet{mylinkcolor}{blue}
\colorlet{mycitecolor}{blue}
\colorlet{myurlcolor}{blue}

\hypersetup{
  linkcolor  = mylinkcolor!\myshade!black,
  citecolor  = mycitecolor!\myshade!black,
  urlcolor   = myurlcolor!\myshade!black,
  colorlinks = true,
}

\newcommand{\abs}{\text{abs}}

\newcommand{\uatl}{\text{UATL}}
\newcommand{\baal}{\text{AAL}}
\newcommand{\ubaal}{\text{UAAL}}
\newcommand{\aarm}{\text{AARM}}
\newcommand{\paarm}{\text{PAARM}}
\newcommand{\uaal}{\text{UAAL}}
\newcommand{\paal}{\text{PAAL}}

\newcommand{\sat}{\text{3SAT}}
\newcommand{\gsat}{\text{SAT}}
\newcommand{\cnfsat}{\text{CNF-SAT}}
\newcommand{\dspace}{\text{DSPACE}}
\newcommand{\dtime}{\text{DTIME}}
\newcommand{\nspace}{\text{NSPACE}}

\newcommand{\ethree}{\text{E3SETS}}
\newcommand{\efour}{\text{E4SETS}}
\newcommand{\p}{\text{P}}
\newcommand{\ap}{\text{AP}}
\newcommand{\ph}{\text{PH}}
\newcommand{\np}{\text{NP}}
\newcommand{\tower}{\text{Tower}}
\newcommand{\tqbf}{\text{TQBF}}
\newcommand{\nll}{\text{NLOGSPACE}}
\newcommand{\ls}{\text{LOGSPACE}}
\newcommand{\pspace}{\text{PSPACE}}
\newcommand{\reg}{\text{REG}}

\newcommand{\bigO}{\mathcal{O}}
\newcommand{\smallO}{$o$}
\newcommand{\cA}{\mathcal{A}}
\newcommand{\cB}{\mathcal{B}}
\newcommand{\cH}{\mathcal{H}}
\newcommand{\ve}{\varepsilon}
\newcommand{\bstring}{\mathit{string}}
\newcommand{\natnum}{\mathbb{N}}

\newtheorem{remark*}[theorem]{Remark}

\begin{document}

\title{Alternating Automatic Register Machines\thanks{Z.~Gao (as RF) and S.~Jain (as Co-PI), F.~Stephan (as PI)
have been supported by the Singapore Ministry
of Education Academic Research Fund grant MOE2019-T2-2-121 / R146-000-304-112.
Furthermore, S.~Jain is supported in part by NUS grants C252-000-087-001
and E-252-00-0021-01.
Further support is 
acknowledged for the NUS tier 1 grants AcRF R146-000-337-114 (F.~Stephan as 
PI) and R252-000-C17-114 (F.~Stephan as PI).}}

\titlerunning{Alternating Automatic Register Machines}

\author{Ziyuan Gao\inst{1} \and
Sanjay Jain\inst{2} \and
Zeyong Li\inst{3} \and
Ammar Fathin Sabili\inst{2} \and
Frank Stephan\inst{1,2}}
\authorrunning{Z. Gao, S. Jain, Z. Li, A. F. Sabili, and F. Stephan}
%
\institute{Department of Mathematics, National University of Singapore, 10 Lower Kent Ridge Road, S17, Singapore 119076, Republic of Singapore \\
\email{matgaoz@nus.edu.sg} \and
School of Computing, National University of Singapore, 13 Computing Drive, COM1, Singapore 117417, Republic of Singapore
\email{\{sanjay,ammar,fstephan\}@comp.nus.edu.sg} \\
\and
Centre for Quantum Technologies, National University of Singapore,
Block S15, 3 Science Drive 2, Singapore 117543 \\
\email{li.zeyong@u.nus.edu}}

\maketitle

\begin{abstract}
This paper introduces and studies a new model of computation called an Alternating Automatic Register 	
Machine (\aarm). An \aarm\ possesses the basic features of a conventional register machine and an alternating Turing machine, but can carry out computations using bounded automatic relations in a single step. One finding is that an \aarm\ can recognise some NP-complete problems, 
including \cnfsat~(using a particular coding), in $\log^* n + \bigO(1)$ steps.  
On the other hand, if all problems in \p~can be solved by an \aarm\ in $\bigO(\log^*n)$ rounds, then
$\p \subset \pspace$.

Furthermore, we study an even more computationally powerful machine, called a Polynomial-Size Padded 
Alternating Automatic Register Machine (\paarm), which allows the input to be padded with a polynomial-size 
string. It is shown that the polynomial hierarchy can be characterised as the languages that are recognised by a \paarm\ in $\log^*n + \bigO(1)$ steps. 
These results illustrate the power of alternation when combined with computations involving 
automatic relations, and uncover a finer gradation between known complexity classes.

	\end{abstract}

	\section{Introduction}
	\label{sec:introduction}

	
	
Automatic structures generalise the notion of regularity for languages to other mathematical objects such as 
functions, relations and groups, and were discovered independently by 
Hodgson \cite{Hodgson76,Hodgson83}, Khoussainov and Nerode \cite{KN95} as well as Blumensath and Gr\"{a}del \cite{Blumensath99,Blumensath00}.  One 
of the original motivations for studying automaticity in general 
structures came from computable structure theory, in particular 
the problem of classifying the isomorphism types of computable 
structures and identifying isomorphism invariants.  In computer 
science, automatic structures arise in the area of infinite state 
model checking; for example, Regular Model Checking, a symbolic framework 
for modelling and verifying infinite-state systems, can be expressed in Existential Second-Order Logic over automatic structures \cite{Lin20}.  Although finite-state transducers are a somewhat more popular extension of ordinary finite-state automata for defining relations between sets of strings, there are several advantages of working with automatic relations, including the following: (1) in general, automatic relations enjoy better decidability properties than finite-state transducers; for example, equivalence between ordinary automata is decidable while this is not so for finite-state transducers; (2) automatic relations are closed under first-order definability \cite{Hodgson83,Khoussainov10,KN95} while finite-state transducers are not closed under certain simple operations such as intersection and complementation.

In this paper, we introduce a new model of computation, called an {\em Alternating Automatic Register Machine (\aarm)}, that is analogous to an alternating Turing machine but may incorporate {\em bounded} automatic relations\footnote{Here an update relation is {\em bounded} if there is a constant such that each possible output is at most that constant longer than the longest input parameter; see Section \ref{sec:preliminaries}.  Since we mainly consider bounded automatic relations in this paper, such relations will occasionally be called ``automatic relations''.} into each computation step.  The main motivation is to try to discover new interesting complexity classes defined via machines where automatic relations are taken as primitive steps, and use them to understand relationships between fundamental complexity classes such as \p,~\pspace~and \np.  More powerful computational models are often obtained
by giving the computing device more workspace or by allowing non-deterministic or {\em alternating} computations, where alternation is a well-known generalisation of non-determinism.  We take up both approaches in this work, extending the notion of alternation to automatic relation computations.  An \aarm\ is similar to a conventional register machine in that it consists of a {\em register} $R$ containing a string over a fixed alphabet at any point in time, and the contents of $R$ may be updated in response to {\em instructions}.  One novel feature of an \aarm\ is that the 
contents of the register can be non-deterministically updated using 
an automatic relation.  Specifically, 
an instruction $J$ is an automatic relation.
Executing the instruction, when the content of the register $R$ is $r$,
means that the contents of $R$ is updated to any $x$ in $\{x: (x,r) \in J\}$;
if there is no such $x$, then the program halts.
Each \aarm\ contains two finite classes, denoted here as $A$ and $B$, 
of instructions; during a computation, instructions are selected alternately 
from $A$ and $B$ and executed.    

To further explain how a computation of an \aarm\ is carried out, we first 
recall the notion of an {\em alternating Turing machine} as formulated by 
Chandra, Kozen and Stockmeyer \cite{Chandra81}.  As mentioned earlier,
alternation is a generalisation of non-determinism, and it is useful for 
understanding the relationships between various complexity classes such as those 
in the polynomial hierarchy (\ph).  The computation of an alternating 
Turing machine can be viewed as a game in which two players -- Anke 
and Boris -- make moves (not necessarily strictly alternating) beginning in 
the start configuration of the machine with a given input $w$ \cite{Furer81}.  
Anke moves from existential configurations (configurations with 
an existential state), and Boris from universal configurations 
(configurations with a universal state) to successor configurations 
according to the machine's transition function.  Anke wins if, after 
a finite number of moves, an accepting configuration is reached.  The 
input $w$ is then accepted by the machine iff Anke has a winning strategy.  
Our definition of an \aarm\ computation is inspired by this game-theoretic 
interpretation of alternating Turing machines.  Given an input $w$, which 
is the string in $R$ at the start of the computation, Anke and Boris move 
in alternating turns, with Anke moving first.  During Anke's turn, she 
carries out any single instruction in the predefined set $A$, possibly 
changing the contents of the register.  Boris moves similarly during his 
turn, except that he carries out an instruction in $B$.  Anke wins if a 
configuration is reached such that Boris is no longer able to carry out 
any instructions in $B$, 
and $w$ is {\em accepted} iff Anke has a winning strategy.     
We also introduce and study Polynomial-Size Padded Bounded Alternating Automatic Register Machines (\paarm s),
which allow a polynomial-size padding to the input of an \aarm.           

The idea of defining computing devices capable of performing single-step operations that are 
more sophisticated than the basic operations of Turing machines is not new.  For example, 
Floyd and Knuth \cite{Floyd90} studied {\em addition machines}, which are finite register 
machines that can carry out addition, subtraction and comparison as primitive steps.  
{\em Unlimited register machines}, introduced by Shepherdson and Sturgis \cite{Shepherdson63}, 
can copy the number in a register to any register in a single step.  Bordihn, Fernau, Holzer, 
Manca and Mart{\'\i}n-Vide \cite{Bordihn06} investigated 
another kind of language generating device called an {\em iterated sequential transducer}, 
whose complexity is usually measured by its number of states (or {\em state complexity}).  
More recently, Kutrib, Malcher, Mereghetti and Palano \cite{Kutrib20} proposed a 
variant of an iterated sequential transducer that performs length-preserving transductions 
on left-to-right sweeps.  Automatic relations are more expressive than arithmetic operations 
such as addition or subtraction, and yet they are not too complex in that even 
one-tape linear-time Turing machines are computationally more powerful; for instance, the 
function that erases all leading 0's in any given binary word can be computed by a one-tape 
Turing machine in linear time but it is not automatic \cite{Stephan15}.  Despite the 
computational limits of automatic relations, we show in Theorem \ref{the:sat_baal} below that the NP-complete Boolean satisfiability problem can be recognised by an \aarm\ in $\log^*n+\bigO(1)$ steps, where $n$ is the length of the formula. 
The results not only show a proof-of-concept for the use of automatic relations in models of 
computation, but also shed new light on the relationships between known complexity classes.

\section{Preliminaries}
	\label{sec:preliminaries}
	
	Let $\Sigma$ denote 
a finite alphabet. We consider set operations including union ($\cup$), concatenation ($\cdot$), Kleene star ($\ast$), intersection ($\cap$) and complement ($\lnot$). Let $\Sigma^*$ denote the set of all strings over $\Sigma$.  A {\em language} is a set of strings. Let the empty string be denoted by $\varepsilon$. For a string $w \in \Sigma^*$, let $|w|$ denote the length of $w$ and $w = w_1w_2...w_{|w|}$ where $w_i \in \Sigma$ denotes the $i$-th symbol of $w$.  Fix a special symbol $\#$ not in $\Sigma$.
Let $x, y \in \Sigma^*$ such that $x = x_1x_2\ldots x_m$ and $y = y_1y_2\ldots y_n$. Let $x' = x_1'x_2'\ldots x_r'$ and $y' = y_1'y_2'\ldots y_r'$ where $r = \max(m,n)$, $x_i' = x_i$ if $i \leq m$ else $\#$, and $y_i' = y_i$ if $i \leq n$ else $\#$. Then, the {\em convolution} of $x$ and $y$ is a string over 
$(\Sigma \cup \{\#\})\times(\Sigma \cup \{\#\})$, defined as $conv(x,y) = (x_1',y_1')(x_2',y_2')\ldots (x_r',y_r')$.   
A relation $J \subseteq X \times Y$ is {\em automatic} if the set $\{conv(x,y): (x,y) \in J\}$ is regular, where the alphabet is $(\Sigma \cup \{\#\}) \times (\Sigma \cup \{\#\})$. Likewise, a function $f : X \rightarrow Y$ is {\em automatic} if the relation $\{(x,y) : x \in domain(f) \land y = f(x)\}$ is automatic  \cite{Stephan}. An automatic relation $J$ is {\em bounded} if $\exists$ constant $c$ such that $\forall (x,y) \in J, \abs(|y| - |x|) \leq c$. On the 
other hand, an unbounded automatic relation has no such restriction. The problem of determining satisfiability of any given Boolean formula in conjunctive normal form will be denoted by $\cnfsat$.  Automatic functions and relations 
	have a particularly nice feature as shown in the following theorem.
	
	\begin{theorem}[\cite{Hodgson83,KN95}]\label{the:first_order_definable}
		Every function or relation which is first-order definable from a finite number of automatic functions and relations is automatic, and the corresponding automaton can be effectively computed from the given automata.
	\end{theorem}

\section{Alternating Automatic Register Machines} 

An {\em Alternating Automatic Register Machine} (\aarm) consists of a {\em register} $R$ and two finite sets $A$ and $B$ of {\em instructions}.  $A$ and $B$ are not necessarily disjoint.  Formally, we denote an \aarm\ by $M$ and represent it as a quadruple $(\Gamma,\Sigma,A,B)$.  (An equivalent model may allow for multiple registers.)  At any point in time, the register contains a string, possibly empty, over a fixed alphabet $\Gamma$ called the {\em register alphabet}.  The current string in $R$ is denoted by $r$.  Initially, $R$ contains an input string over $\Sigma$, an {\em input alphabet} with $\Sigma \subseteq \Gamma$.  Strings over $\Sigma$ will sometimes be called {\em words}.  The contents of the register may be changed in response to an instruction.  
An instruction $J \subseteq \Gamma^* \times \Gamma^* $ is a bounded automatic relation; this changes the contents of $R$ to some $x$ such that $(r,x) \in J$ (if such an $x$ exists).  The instructions in $A$ and $B$ are labelled $I_1,I_2,\ldots,$ (in no particular order and not necessarily distinct).  A {\em configuration} is a triple $(\ell,r,w)$, where $I_{\ell}$ is the current instruction's label and $r,w \in \Gamma^*$.  Instructions are generally nondeterministic, that is, there may be more than one way in which the string in $R$ is changed from a given configuration in response to an instruction.  A {\em computation history} of an \aarm\ {\em with input $w$} for any $w \in \Sigma^*$ is a finite or infinite sequence $c_1,c_2,c_3,\ldots$ of configurations such that the following conditions hold.  Let $c_i = (\ell_i,r_i,w_i)$ for all $i$.
\begin{itemize}
\item $r_1 = w$.  We call $c_1$ the {\em initial configuration} of the computation history.
\item For all $i$, $(r_i,w_i) \in I_{\ell_i}$.  This means   
  that $I_{\ell_i}$ can be carried out using the current register contents, 
changing the contents of $R$ to   
  $w_i$. 
\item Instructions executed at odd terms of the sequence belong to $A$, while those executed at even terms   
  belong to $B$: 

\begin{quote}
\begin{center}
$I_{\ell_i} \in \left\{\begin{array}{ll}
A & \mbox{if $i$ is odd;} \\
B & \mbox{if $i$ is even.}\end{array}\right.$    
\end{center}
\end{quote}

\item If $c_{i+1}$ is defined, then $r_{i+1} = w_i$.
  In other words, the contents of $R$ are (nondeterministically) updated 
  according to the instruction and register   
  contents of the previous configuration.      
\item Suppose $i$ is odd (resp.~even) and $c_i=(I_{\ell_i},r_i,w_i)$ 
is defined.  
If there is some $I_{\ell} \in B$ (resp.~$I_{\ell} \in A$)
with $\{x: (w_i,x) \in I_{\ell}\}$ nonempty, 
then $c_{i+1}$ is defined.  
In other words, the computation continues so long as it is 
possible to execute an instruction from the appropriate set, 
either $A$ or $B$, at the current term.      
\end{itemize}

\noindent
We interpret a computation history of an \aarm\ as a sequential game between two players, Anke and Boris, where Anke moves during odd turns and Boris moves during even turns.  During Anke's turn, she must pick some instruction 
$J$ from $A$ such that $\{x: (r,x) \in J\}$ is nonempty and select some $w \in \{x: (r,x) \in J\}$; if no such instruction exists, then the game terminates.  The contents of $R$ are then changed to $w$ at the start of the next turn.  The moving rules for Boris are defined analogously, except that he must pick instructions only from $B$.  Anke {\em wins} if the game terminates after a finite number of turns and she is the last player to execute an instruction; in other words, Boris is no longer able to carry out an instruction in $B$ and the {\em length} of the game (or computation history), measured by the total number of turns up to and including the last turn, is odd.  The condition for Boris to win is defined symmetrically. 
A draw game is one such that no player wins in finitely many turns, that is, the game runs forever. 
The \aarm\ {\em accepts} a word $w$ if Anke can move in such a way that she will always win a game with an initial configuration $(\ell,w,v)$ for some $I_{\ell} \in A$ and 
$v \in \Gamma^*$, regardless of how Boris moves.  To state this acceptance condition more formally, one could define Anke's and Boris' {\em strategies} to be functions $\mathcal{A}$ and $\mathcal{B}$ respectively with 
$\mathcal{A}:(\natnum \times \Gamma^* \times \Gamma^*)^* \times \Gamma^* \mapsto A \times \Gamma^*$ and 
$\mathcal{B}:(\natnum \times \Gamma^* \times \Gamma^*)^* \times \Gamma^* \mapsto B \times \Gamma^*$, which map each segment of a computation history together with the current contents of $R$ to a pair specifying an instruction as well as the new contents of $R$ at the start of the next round according to the moving rules given earlier.  The \aarm\ accepts $w$ if there is an $\mathcal{A}$ such that for every 
$\mathcal{B}$, there is a finite computation history $\langle c_1,\ldots,c_{2n+1}\rangle$ where
\begin{itemize}
\item $c_i = (\ell_i,r_i,w_i)$ for each $i$, 
\item $r_1 = w$, 
\item $\mathcal{A}((\langle c_{i}: i < 2j+1 \rangle,r_{2j+1})) = (I_{\ell_{2j+1}},w_{2j+1})$ for each 
  $j \in \{0,\ldots,n\}$,
\item $\mathcal{B}((\langle c_{i}: i < 2k \rangle,r_{2k})) = (I_{\ell_{2k}},w_{2k})$ for each 
  $k \in \{1,\ldots,n\}$; 
\item there is no move for $B$ in $c_{2n+1}$, that is 
no instruction in $B$ contains a pair of the form $(w_{2n+1},\cdot)$.
\end{itemize}
Here $\langle c_{i}: i < k\rangle$ denotes the sequence $\langle c_1,\ldots,c_{k-1}\rangle$, which is empty if 
$k \leq 1$.  Such an $\mathcal{A}$ is called a {\em winning strategy for Anke with respect to $(M,w)$}.  Given a winning strategy $\mathcal{A}$ for Anke with respect to $(M,w)$ and any strategy $\mathcal{B}$, the corresponding computation history of $M$ with input $w$ is unique and will be denoted by 
$\mathcal{H}(\mathcal{A},\mathcal{B},M,w)$.  In most subsequent proofs, $\mathcal{A}$ and $\mathcal{B}$ will generally not be defined so formally.  Set 

\begin{quote}
\begin{center}
$L(M) := \{w \in \Sigma^*: \mbox{$M$ accepts $w$}\}$;
\end{center}
\end{quote}

\noindent  
one says that $M$ {\em recognises} $L(M)$.       
Note that even though we have given the description of
\aarm~via alternation of moves by Anke and Boris, 
it is possible to define games where strict alternation is not needed.
Furthermore, a constant amount of extra state information can be 
stored in the register.  

\begin{definition}[Alternating Automatic Register Machine Complexity]\label{defn:aal}
Let $M = (\Gamma,\Sigma,A,B)$ be an \aarm\ and let $t \in \natnum_0$.  
For each $w \in \Sigma^*$, $M$ {\em accepts $w$ in time $t$} if Anke has a winning
strategy $\mathcal{A}$ with respect to $(M,w)$ such that for any strategy $\mathcal{B}$
played by Boris, the length of $\cH(\cA,\cB,M,w)$ is not more than $t$.  
(As defined earlier, $\cH(\cA,\cB,M,w)$ is the computation history of $M$ with input $w$ 
when $\cA$ and $\cB$ are applied.)  

An \aarm\ decides a language $L$ in $f(n)$ steps for a function $f$ depending on the length $n$ of the input if for all $x \in \{0,1\}^n$, both players can enforce that the game terminates within $f(n)$ steps by playing optimally and one player has a winning strategy needing at most $f(n)$ moves and $x \in L$ if Anke is the player with the winning strategy. 
$\baal[f(n)]$ denotes the family of languages decided by \aarm s that decide in time $f(n)$.  

\end{definition}

We begin with several
examples to illustrate \aarm s and how they carry out computations.

\begin{example}\label{exmp:aal}
\begin{enumerate}
\item If $A$ and $B$ are both empty, then every computation history of the corresponding \aarm\ is empty; thus 
  the \aarm\ does not accept any input.   

\item Suppose $\Gamma = \{0,1\}$.  Let $A$ consist of the single instruction 
  $I_1$ which consists of 
  all pairs $(y,x) \in \Gamma^* \times \Gamma^*$ with $x = y$, and 
  suppose $B$ is empty.  
It is easy to see that $I_1$ is an automatic relation.  
Since $\{x: (v,x) \in I_1\}$ is nonempty for any $v$ and $B$ is empty, 
it follows that every  computation history with any input $w$ always ends at 
the initial configuration $(1,w,w)$. 
The corresponding \aarm\ thus accepts {\em any} binary word. 
          
\item Suppose $\Gamma = \{0,1,\overset{\bullet}{0},\overset{\bullet}{1},\blacksquare,\bigstar\}$.  
  $\overset{\bullet}{i}$ is a ``marked'' version of $i$ for $i \in \{0,1\}$; $\blacksquare$ is a special symbol   
  indicating a rejection of the input, while $\bigstar$ indicates an acceptance of the input.  Let $A$ consist of the   
  instructions $I_1$ and $I_2$.
  $I_1 \subseteq \{0,1,\overset{\bullet}{0},  
  \overset{\bullet}{1}\}^* \times \{0,1,\overset{\bullet}{0},\overset{\bullet}{1}\}^*$ is the relation containing all pairs $(y,x)$ such that 
  $y \in \{0,\overset{\bullet}{0}\}^*\cdot\{1,\overset{\bullet}{1}\}^*$ and $x$ is the string obtained from $y$   
  by marking every alternate unmarked $0$ starting with the first unmarked $0$ 
  as well as every alternate unmarked  
  $1$ starting with the first unmarked $1$, while leaving other symbols unchanged.  For example, 
$
(0  0  0  1  1, \overset{\bullet}{0}0\overset{\bullet}{0}\overset{\bullet}{1}1
)   \in I_1,
(\overset{\bullet}{0} 0 \overset{\bullet}{0} \overset{\bullet}{1} 1 \times \overset{\bullet}{0} 0 \overset{\bullet}{0} \overset{\bullet}{1} 
\overset{\bullet}{1})  \notin I_1.
$
%

$I_2 \subseteq \{\overset{\bullet}{0},\overset{\bullet}{1},\bigstar\}^* \times \Gamma^*$ contains all 
  pairs $(y,x)$ such that $y$ does not have any unmarked $0$'s or $1$'s and $x = y \cdot \bigstar$. 

  Let $B$ consist of the instruction $I_3$, where
  $I_3 \subseteq \{0,1,\overset{\bullet}{0},\overset{\bullet}{1}\}^* \times \Gamma^*$
  is the relation containing all pairs $(y,x)$ such that $x = y \cdot \blacksquare$ if the total number of unmarked   
  $0$'s and $1$'s in $y$ is odd and $x = y$ otherwise.  
  One can show as before that $I_1,I_2$ and $I_3$ are automatic.  Let $w \in \{0,1\}^*$ be any given input.  If $0$ 
  occurs to the right of some $1$ in $w$, then each computation history with input $w$ is empty, so the \aarm\ 
  does not accept $w$.  If $w$ is of the shape $0^n 1^m$ with $n \neq m$, then, after an odd number of   
  turns, the total number of unmarked $0$'s or $1$'s will be odd, and executing $I_3$ from $B$ during the next 
  turn will result in $R$ containing a string with the symbol $\blacksquare$.  No instructions in $A$ can 
  subsequently be carried out.  If $w$ is of the shape $0^n 1^n$, then eventually all the symbols in $w$ will be 
  marked and $I_2$ can be carried out, resulting in $R$ containing a string with the symbol $\bigstar$.  No 
  instructions in $B$ can then be executed.  The corresponding \aarm\ thus accepts the language of all binary 
  words of the shape $0^n1^n$.  
  
\end{enumerate}
\end{example}

\noindent
We first establish the closure of $\baal[\bigO(f(n))]$ under the usual set-theoretic Boolean operations as well as the regular operations.

\begin{theorem}\label{thm:closurebaal}
		Let $\star$ denote any one of the operations union, intersection and concatenation.
		For each $L_1 \in \baal[\bigO(f(n))]$ and each $L_2 \in \baal[g(n)]$, $L_1 \star L_2 \in \baal[\max\{f(n),g(n)\}]$
		and $L_1^* \in \baal[\bigO(f(n))]$.
		In particular, the languages in $\baal[\bigO(f(n))]$ 
		are closed under union, intersection, concatenation 
	    and Kleene star. 

	\end{theorem}

\begin{proof}
Suppose $L_1 \in \baal[f(n)]$ and $L_2 \in \baal[g(n)]$.  Let $M_1 = (\Gamma,\Sigma,A,B)$
		and $M_2 = (\Gamma',\Sigma,A',B')$ be \aarm s such that $L_1 = L(M_1)$, $L_2 = L(M_2)$
		and $M_1$ and $M_2$ recognise in times $f(n)$ and $g(n)$ respectively.  We provide the formal details
		for constructing an \aarm\ $M_3 = (\Gamma'',\Sigma,A'',B'')$ recognising $L_1 \cup L_2$ in time 
		$\max\{f(n),g(n)\}$, and give informal descriptions of the \aarm s for the other operations.

		\begin{itemize}
\setlength\itemsep{.5em}

		\item The components of $M_3$ are defined as	 follows.  Suppose $\boxplus \notin \Gamma \cup \Gamma'$
and $\lozenge \notin \Gamma \cup \Gamma'$.  

\begin{enumerate}[label=(\roman*)]

\item $\Gamma'' = \Gamma \cup \Gamma' \cup \{\boxplus,\lozenge\}$.

\item\label{step1} Let $J$ (resp.~$J'$) be the relation consisting of all pairs $(x,\boxplus x)$ 
  (resp.~$(x,\linebreak[3]\lozenge x)$) with $x \in \Sigma^*$ and put 
$J$ 
  (resp.~$J'$) into $A''$.

\item\label{step2} Let $K$ (resp.~$K'$) be the relation consisting of all pairs $(x,\boxplus x,)$ 
  (resp.~$(x,\linebreak[3]\lozenge x)$) with $x \in \boxplus \cdot \Sigma^*$ (resp.~$x \in \lozenge \cdot \Sigma^*$) 
  and put $K$ (resp.~$K'$) into $B''$.

\item\label{step3} For every instruction $L$ in $A$ (resp.~$B$), let $L'$ be 
the   
  relation such that $(\boxplus \boxplus y, \boxplus \boxplus x) \in L' \Leftrightarrow (y,x) \in L$, and put 
  $L'$ into $A''$ (resp.~$B''$).

\item\label{step4} For every instruction $L$ in $A'$ (resp.~$B'$), let $L'$ be the relation   
  such that $(\lozenge \lozenge y, \lozenge \lozenge x) \in L' \Leftrightarrow (y,x) \in L$, and put $L'$ into $A''$ (resp.~$B''$).

\end{enumerate}    

		Given any $w \in L_1 \cup L_2$, a winning strategy for Anke with respect to $(M_3,w)$ is as follows.  At the start of the game, Anke determines whether $w \in L_1$ or $w \in L_2$ (if $w \in L_1 \cap L_2$, then Anke can just follow the subsequent steps for the case $w \in L_1$).  Assume $w \in L_1$.  The first move of Anke is to change the contents of $R$ from $w$ to $\boxplus w$ using one of the instructions defined in \ref{step1}.  During the second turn, the only move Boris can make is to change the current contents of $R$ by prepending it with $\boxplus$, using one of the instructions defined in \ref{step2}.  The third turn starts with $\boxplus\boxplus w$ in $R$.  Anke then adapts a winning strategy with respect to $(M_1,w)$, using the instructions defined in \ref{step3}; the only difference between the current strategy and that for $(M_1,w)$ is that the current one applies to the register's contents prepended with $\boxplus\boxplus$.  From the third turn onwards, Anke and Boris can only apply instructions defined in \ref{step3} and so the game would proceed as in a computation history of $M_1$ with input $w$ (except that the register's contents are prepended with $\boxplus\boxplus$).  In the case that $w \in L_2$, Anke proceeds similarly, except that the third turn starts with $\lozenge\lozenge w$ in $R$ and she would adapt a winning strategy with respect to $(M_2,w)$.  In either case, the length of the game is 
$\bigO(\max\{f(|w|),g(|w|)\})+2 = \bigO(\max\{f(|w|),g(|w|)\})$.  


Suppose $w \notin L_1 \cup L_2$.  If $R$ contains $\boxplus\boxplus w$ (resp.~$\lozenge\lozenge w$) at the start of the third turn, then the game would proceed in a way similar to that of a computation history of $M_1$(resp.~$M_2$) with input $w$, that is, either it terminates with Boris making the last move or it never terminates.  Therefore $L(M_3) = L_1 \cup L_2$ and $M_3$ recognises in time $\bigO(\max\{f(n),g(n)\})$.  

	\item For $L_1 \cap L_2$, a technique similar to that in (i) can be applied.  $A''$ will have an instruction allowing Anke to skip the first turn and $B''$ will have an instruction allowing Boris in the second turn to ``pick'' $L_1$ or $L_2$ (say by prepending the input with a special symbol).  If Boris picks $L_1$, then the remainder of the game would proceed as in a computation history of $M_1$ with input $w$; otherwise, it proceeds as in a computation history of $M_2$ with input $w$.   

	\item For $L_1 \cdot L_2$, 
Anke would first mark a splitting point 
in the input, then 
Boris will pick which side he wants to verify.    

	\item For $L_1^*$, 
Anke would first mark all splitting points, then 
Boris will pick a single string he wants to verify.
\end{itemize}
\end{proof}

\noindent
We now show that the family of all languages recognised by some \aarm\ coincides with the family of all recursively enumerable (r.e.) languages.  First, it is shown that any conventional register machine can be simulated with an \aarm.
We use a slight variant of Shepherdson and Sturgis' \cite{Shepherdson63} unlimited register machine. 

\begin{proposition}\label{prop:simregmachine}
Consider any register machine $M'$ with $n$ registers $R_1,\ldots,R_n$ and a program
consisting of a sequence $I_1,\ldots,I_k$ of instructions, each taking one of the following forms:
\begin{itemize}
\item $R^+ \rightarrow I'$, which adds $1$ to the contents of register $R$ and jumps to instruction $I'$;
\item $R^- \rightarrow I',I''$, which subtracts $1$ from the contents of register $R$ and jumps to instruction 
  $I'$ if the current number in $R$ is positive, else jumps to instruction $I''$; 
\item \textsc{HALT}, which stops any further instructions from being carried out. 
\end{itemize}
The contents of the registers at any point in time are binary strings, each representing a natural number.
Then there is an \aarm\ $M$ such that for each binary word $w$, $M$ accepts $w$ iff $M'$ halts when $R_1$ initially contains $w$ and $R_i$ initially contains $0$ for each $i > 1$.
\end{proposition}    

\def\proofofsimregmachine{
\begin{proof}
For any natural number $k$, let $\bstring(k)$ denote the binary representation of $k$.  Let $M$ be an \aarm\ with register $R$ and instruction sets $A$ and $B$.  Let $\Gamma = \{0,1,\#,\bigstar\}$.  $R$ will simulate the current instruction being read by $M'$ as well as the current contents of $R_1,\ldots,R_n$; each of these simulated items will be separated with the symbol $\#$.  Suppose $R_1$ initially contains $w$ and $R_i$ initially contains $0$ for each $i > 1$.  We may assume without loss of generality that $w \neq \varepsilon$ and $R$ initially contains the string 

\begin{quote}
\begin{center}
$\bstring(1)\# w \# 0 \# \ldots \# 0$,  
\end{center}
\end{quote}

\noindent     
which contains $n$ occurrences of $\#$.  For each $I_i$, $A$ will contain a corresponding instruction $I'_i$ defined according to the following case distinction:

\begin{description}
\setlength\itemsep{1em}

\item[$I_i := R_j^+ \rightarrow I_k$:] Let $I_i'$ be the relation 
consisting of all pairs $(y,x)$ such that $y$ is 
  of the form $\bstring(i) \# \bstring(i_1) \# \bstring(i_2) \# \ldots \# \bstring(i_n)$ and $x$ equals 
  $\bstring(k)\# \bstring(i_1) \# \ldots$ $\# \bstring(i_j+1) \# \ldots \# \bstring(i_n)$.
This instruction simulates, when running $M'$, the addition of $1$   
  to $R_j$ and the jump from $I_i$ to $I_k$.

\item[$I_i := R_j^- \rightarrow I_k,I_{\ell}$:] Let $I_i'$ be the relation consisting of all pairs $(y,x)$ such that
  $y$ is of the form $\bstring(i) \# \bstring(i_1) \# \bstring(i_2) \# \ldots \# \bstring(i_n)$; if $y$ is of this   
  form, then $i_j = 0$ implies $x$ equals $\bstring(\ell) \# \bstring(i_1) \# \ldots \# \bstring(i_n)$, that is, only 
  the first segment is changed and it becomes $\bstring(\ell)$, and $i_j > 0$ implies $x$ equals 
  $\bstring(k) \# \bstring(i_1) \# \ldots$ $\bstring(i_j-1) \# \ldots \# \bstring(i_n)$, that is, the first segment is   
  changed to $\bstring(k)$ and the $(j+1)$-st segment is changed to $\bstring(i_j-1)$.  
If $R_j$ currently contains a positive number, then $I'_i$ simulates the  
  subtraction of $1$ from the number currently in $R_j$ and the jump from $I_i$ to $I_k$; otherwise, $I'_i$  
  simulates the jump from $I_i$ to $I_{\ell}$.      

\item[$I_i := \textsc{HALT}$:] Let $I_i'$ be the relation consisting of all pairs $(y,x)$ such that $y$ is of the 
  form $\bstring(i) \# \bstring(i_1) \# \bstring(i_2) \# \ldots \# \bstring(i_n)$ and $x$ equals $\bigstar^{|y|}$.  
$I'_i$ changes the contents of $R$ to a nonempty string over 
  $\{\bigstar\}$, signalling that $M'$ has halted.      
\end{description}

$B$ consists of the single instruction $R' \in \{x: (R',x) \in N\}$, where 
$N \subseteq (\Gamma \setminus \{\bigstar\})^* \times \Gamma^*$ is the relation consisting of all pairs 
$(y,x)$ such that $x = y$.  Suppose $M'$ eventually halts when $R_1$ initially contains $w$ and $R_i$ contains $0$ for each $i > 1$.  Anke carries out instruction $I'_i$ whenever $I_i$ applies to $M'$, starting with $i = 1$.  Boris can only carry out one instruction, which leaves the contents of $R'$ unchanged.  $M'$ halts only when some $I_i := \textsc{HALT}$ is applied; Anke then applies $I'_i$, which changes the contents of $R'$ to a nonempty string over $\{\bigstar\}$.  Since Boris cannot execute any instruction when $R'$ contains a string with an occurrence of $\bigstar$, the game will always end with Anke.  On the other hand, if $M'$ never halts, then the contents of $R'$ will never change to a string containing $\bigstar$ and so Boris can always keep the game going.  Thus $M$ accepts $w$ iff $M'$ halts.
\end{proof}
}

\proofofsimregmachine

\noindent
It remains to show that an \aarm\ recognises only r.e.\ languages.  The proof of this fact relies on the boundedness assumption for automatic relations used in register updates of an \aarm.  

\begin{proposition}\label{prop:baalre}
For each \aarm\ $M$, $L(M)$ is r.e.
\end{proposition}

\def\proofofbaalre{
\begin{proof}
Suppose $w$ is given as an input to $M$.  The family of all computation histories of $M$ with input $w$ may be viewed as a tree $\tau$ consisting of a root node $(0,\ve,\ve)$ and all possible configurations (each represented as a node) of such a computation history.  
Each maximal path in $\tau$ starting with the root node is a concatenation of the root node and a computation history of $M$ with input $w$. 
Owing to the boundedness of all automatic relations used in the instructions of $M$, $\tau$ is finitely branching.  Suppose Anke has a winning strategy with respect to $(M,w)$.  Consider the subtree $\tau'$ of all computation histories prepended with the root node $(0,\ve,\ve)$ such that Anke always moves according to a fixed winning strategy.  If $\tau'$ were infinite, then by K\"{o}nig's lemma, $\tau'$ would have an infinite branch and Boris would win the game represented by this branch, contradicting the assumption that Anke is applying a winning strategy.  Thus there is some $d$ such that all computation histories when Anke is applying a winning strategy have length at most $d$.  An algorithm can guess the value of $d$ and, 
by Theorem \ref{the:first_order_definable}, determine recursively that Boris can no longer make a valid move when the length of the game has reached $d$ no matter how he plays, and accept $w$.  On the other hand, if Anke does not have a winning strategy, then no such $d$ exists and so the algorithm would never accept $w$.  Such an algorithm would therefore recognise the set of all words for which Anke has a winning strategy.              
\end{proof}
}

\proofofbaalre

\noindent
Propositions \ref{prop:simregmachine} and \ref{prop:baalre} together imply the desired result.

\begin{theorem}\label{thm:baalre}
The family of languages recognised by \aarm s is precisely the family of r.e.\ languages. 
\end{theorem}

\medskip
\noindent
We recall that an alternating Turing machine that decides in $\bigO(f(n))$
time can be simulated by a deterministic Turing machine using $\bigO(f(n))$ 
space.  The following theorem gives a similar connection between the
time complexity of \aarm s and the space complexity of deterministic
Turing machines.

\begin{theorem}\label{thm:baaltimedspace}
For any $f$ such that $f(n) \geq n$, $\baal[\bigO(f(n))] \subseteq \dspace[\bigO((n+f(n))f(n))]=\dspace[\bigO(f(n)^2)]$.
\end{theorem}

\begin{proof}
Given an \aarm\ $M$, there is a constant $c$ such that each register update by an automatic relation used to define an instruction of $M$ increases the length of the register's contents by at most $c$.  Thus, after 
$\bigO(f(n))$ steps, the length of the register's contents is $\bigO(n+f(n))$.  
As implied by \cite[Theorem 2.4]{Case13}, each computation of an automatic relation with an input of length $\bigO(n+f(n))$ can be simulated by a nondeterministic Turing machine in $\bigO(n+f(n))$ steps; this machine can then be converted to a deterministic space $\bigO(n+f(n))$ Turing machine.  If $M$ accepts an input $w$, then there are $\bigO(f(n))$ register updates by automatic relations when Anke applies a winning strategy, and so there is a deterministic Turing machine simulating $M$'s computation with input $w$ using space $\bigO((n+f(n))f(n))$.    
\end{proof}

\noindent 
As a consequence, one obtains the following analogue of the equality
between $\ap$ (classes of languages that are decided by alternating
polynomial time) and $\pspace$.

\begin{corollary}\label{cor:aalpolyequalspspace}
$\bigcup_k \baal[\bigO(n^k)] = \pspace$.
\end{corollary}

\begin{proof}
The containment relation $\baal[\bigO(f(n))] \subseteq \dspace[\bigO((n+f(n))f(n))]$ in Theorem \ref{thm:baaltimedspace} holds whether or not the condition $f(n) \geq n$ holds.
Furthermore, the computation of an alternating Turing machine can be simulated using an 
\aarm, where the transitions from existential (respectively, universal) states correspond to the instructions for Anke (respectively, Boris), and each computation step of the alternating 
Turing machine corresponds to a move by either player.  Therefore 
\pspace, which is equal to \ap, 
is contained in $\bigcup_k \baal[n^k]$.
\end{proof}

\noindent
We come next to a somewhat surprising result: 
an \aarm-program 
can recognise $\sat$ 
using just $\log^*n + \bigO(1)$ steps.  To prove the theorem, we give the 
following lemma, which 
illustrates most of the power of $\aarm$s.

\begin{lemma}[Log-Star Lemma]\label{lem:logstar1}
		Let $u,v \in \Sigma^*$. Let $\#,\$ \not\in \Sigma$. 
Then both languages $\{u'\$v' : 
u' \in \#^*u\#^*, v' \in \#^*v\#^*$ and
$u = v\}$ and $\{u'\$v' : 
u' \in \#^*u\#^*, v' \in \#^*v\#^*$ and
$u \neq v\}$ are in $\baal[\log^* n + \bigO(1)]$.
	\end{lemma}
	
	\noindent
	The Log-Star Lemma essentially states that a comparison of two substrings can be done by
an \aarm\ in $\log^*n + \bigO(1)$ time. This is done by ignoring the unnecessary symbols in the 
register by replacing them with $\#$'s and adding a separator ($\$$) between the two strings.

\begin{proof}
We now prove the Log-Star Lemma.
The algorithm below recursively reduces the problem
to smaller sizes of $u$, $v$ in constant number of steps
(the maximum of the length of $u,v$ is reduced logarithmically in
constant number of steps). 
For the base case, if size of $u$ or $v$ is bounded by a constant,
then clearly both languages can be recognized in one step.

For larger size $u, v$, the algorithm/protocol works as follows.
For ease of explanation, suppose Anke is trying to show that
$u=v$ (case of Anke showing $u \neq v$ will be similar).
Given input $s = u'\$v'$, player Anke will try to give each symbol except \$'s a mark $\in \{0,1,2,3\}$ as follows: 
		
		\begin{enumerate}
			\item For each $\#$ in $u'$ and $v'$, the mark of $3$ will be given.
			\item For the contiguous symbols of $u$, starting from the first symbol, the following infinite marking will be given (whitespaces are for the ease of readability and not part of marking):
$$
20~21~200~201~210~211~2000~2001~2010~2011~2100~2101~2110~2111~20000~\cdots
$$
\noindent
Namely, a series of blocks of string in ascending length-lexicographical order.  Let $T$ be the so defined infinite sequence.
\noindent
Given a string 
$s = u'\$v'$, where $u' \in \#^*u\#^*$ and $v' \in \#^*v\#^*$
for some $u,v \in \Sigma^*$, each contiguous subsequence of $u$ (resp.~$v$) whose sequence of positions is equal to the sequence of positions of $T$ of some string in $2\{0,1\}^*$ such
that the next symbol in $T$ is $2$ will be called a {\em block}.  Each block starts with $2$ followed by a binary string.  Let $k$ ($\geq 2$) be the maximum size of a block.  Summing the lengths of the blocks of $u$ gives that $(k-2)\cdot 2^{k-1} + k \leq n$, and thus $k \leq \log n$. 
			\item For the contiguous symbols of $v$, the marking will be similar.
		\end{enumerate}
		\noindent
		The marking is considered \textit{valid} if all above rules are satisfied. This is 
an example of a \textit{valid} marking of $S = 
		``\#foobar\#\#\$foobar\#\#"$:
		
		\begin{center}
			\begin{tabular}{ c || c c c c c c c c c c c c c c c c c c }
				s & \# & f & o & o & b & a & r & \# & \# & \$ & f & o & o & b & a & r & \# & \#\\ 
				Mark & 3 & 2 & 0 & 2 & 1 & 2 & 0 & 3 & 3 & \$ & 2 & 0 & 2 & 1 & 2 & 0 & 3 & 3   
			\end{tabular}
		\end{center}
		
		\noindent
		If $u = v$ and the marking is \textit{valid}, player Anke will guarantee that each 
symbol of $u$ and $v$ will be marked 
with exactly the same marking. However if $u \neq v$ and the marking is \textit{valid}, either the length of $u$ and $v$ are different or there will be at least a single block which differs on at least one symbol between $u$ and $v$. Therefore, player Boris can have the following choices of challenges:
		
		\begin{enumerate}
			\item Challenge that player Anke did not make a valid marking, or
			\item $|u| \neq |v|$, or
			\item the string in a specified block differs on at least one symbol between $u$ and $v$.
		\end{enumerate}
		
		\noindent
		Notice that $u = v$ if and only if player Boris could not 
successfully challenge player Anke. The first challenge will ensure that player Anke gave a valid marking. There are three possible cases of invalid marking:
		
		\begin{enumerate}[label=(\alph*)]
			\item There is a $\#$ in $u'$ or $v'$ which is not given by a mark of $3$. In this case, player Boris may point out its exact position. Here, player Boris needs $1$ step.
			\item For $u$ and $v$, the first block is not marked with "20". This can also be easily pointed out by player Boris. Here, player Boris needs $1$ step.
			\item For $u$ and $v$, a block is not followed by 
its successor. This can be pointed out by player Boris by checking two 
things: the length of the `successor' block should be less than or equal 
to the length of the `predecessor' block plus one, and the 
`successor' block indeed should be the successor of the `predecessor' block. 
Also, 
we note that the last block may be incomplete.
			\begin{enumerate}[label=(\roman*)]
			\item The length case can be checked by looking at how many symbols there are between the pair of $2$'s bordering each block. Let $p$ and $q$ be the length of `predecessor'
block and the `successor' block respectively. In 
the case that the 
`successor' block is not the last block (not incomplete), player Boris may 
challenge if $p \neq q$ and $p+1 \neq q$. This can be done by marking both blocks with $1$ separated by \$ and the rest
with dummy symbols $\#$ and then doing the protocol for equality of
the modified $u$ and $v$ recursively.
As player Boris may try to find the `short' challenge, 
player Boris will find the earliest block which has the issue 
and thus make sure that $p$ is at most 
logarithmic in the maximum of the lengths of $u$ and $v$. 
As $q$ may be much larger than $p$, player Boris may limit the second 
block by taking at most $p+2$ symbols. 
				\item The successor case can be checked by the following observation. A successor of a binary string can be calculated by finding the last $0$ symbol and flipping all digits from that position to the end while maintaining the previous digits. As an illustration, the successor of "10110\underline{0}111" will be "10110\underline{1}000" where the symbols are separated in $3$ parts: the prefix which is the same, the last $0$ digit, which is underlined, becoming $1$; and all $1$ digits on suffix becoming $0$. 
Player Boris then may challenge the first part not to be equal or the last part not to be the same length or not all $1$'s by providing the 
position of the last $0$ on the `predecessor' block (or the last $1$ on the 
`successor' block, if any).
Checking the equality of two strings can be done recursively, also 
similarly applied for checking the length. Notice a corner case of 
all $1$'s which has the successor consisting of $1$ followed by $0$'s with 
the same length, which can be handled separately. Also 
notice that the `successor' block may be incomplete if it is the last block,
which can also be handled in a similar manner as above. 
			\end{enumerate}
		\end{enumerate}
		
		\noindent
		For the second challenge, player Boris can 
(assuming the marking is valid) check whether the last two blocks of 
$u$ and $v$ are equal. Again, player Boris may limit it for a `short' challenge so the checking size is decreasing to its $\log$. 
For the third challenge, player Boris will specify the two blocks 
on $u$ and $v$ (same block on both) which differ on at least one symbol 
between $u$ and $v$. Again, same protocol will apply and the size 
is decreasing to its $\log$.  Furthermore, both the marking and selection
of blocks are done in a single turn. 

Thus, the above algorithm using one alternation of each of the players
reduces the problem to logarithmic in the size of the maximum of the
lengths of $u$ and $v$.
In particular, when the size of $u$ and $v$ are small enough,
the checking will be done in constant number of steps. 
Thus, the complexity of the problem satisfies:
$T(m_{k+1}) \leq T(\log m_k)$, where $m_i$ denotes the maximum of the sizes
of $u$ and $v$ at step $i$. 
As the lengths of $u,v$ at each step are bounded by the length of the whole
input string, the lemma follows. 
Note that either player can enforce that the algorithm runs in $\log^*n+\bigO(1)$ steps. The player makes the own markings always correct and challenges incorrect markings of the opponent at the first error so that the logarithmic size descent is guaranteed. Challenged correct markings always cause the size to go down once in a logarithmic scale. 
\end{proof}

\begin{theorem} \label{the:sat_baal}
		There is an $\np$-complete problem in $\baal[\log^*n + \bigO(1)]$.
	\end{theorem}

\begin{proof}
		Consider any encoding of a $\gsat$ formula in conjunctive normal form such that after each variable occurrence there is a space for a symbol indicating the truth value of that variable.  For example, literals may be represented as $+$ or $-$ followed by a variable name and then a space for the variable's truth value, clauses may be separated by semicolons, literals may be separated by commas and a dot represents the end of the formula.  Anke sets a truth value for each variable occurrence in the formula and a dfa then checks whether or not between any two semicolons, before the first semicolon and after the last semicolon there is a true literal; if so, Boris can challenge that two identical variables received different truth values.      
It is now player Anke's job to prove that the two variables picked by Boris are different. By the Log-Star Lemma, this verification needs $\log^*n + \bigO(1)$ steps. Hence, $\cnfsat \in \baal[\log^*n + \bigO(1)]$.
	\end{proof}

\noindent
The Log-Star Lemma can also be applied, using a technique similar to that
in the proof of Theorem \ref{the:sat_baal}, to show that for any $k \geq 3$, 
the $\np$-complete problem \textit{$k$-COLOUR} of deciding whether any given
graph $G$ is colourable with $k$ colours belongs to $\baal[\log^*n + \bigO(1)]$.
Using a suitable encoding of nodes, edges and colours as strings, Anke 
first nondeterministically assigns any one of $k$ colours to each node
and ensures that no two adjacent nodes are assigned the same colour;
Boris then challenges Anke on whether there are two substrings of the
current input that encode the same node but encode different colours.

\medskip
\noindent
The next theorem shows that the class $\baal[\log^*n + \bigO(1)]$ contains  \\ \nll.

\begin{theorem}\label{thm:nlsinaal}
$\nll \subseteq \baal[\log^*n + \bigO(1)]$.
\end{theorem}

\begin{proof}
Consider an \nll\ computation that takes time $n^c$.
One splits the input into $\sqrt{n}$ equal-sized blocks, each of which
represents the sequence of configurations of the Turing Machine during
an interval of $\frac{s}{\sqrt{n}}$ steps, where $s$ is the total number of
steps on the input, and
Anke guesses for each block the following information: 
\begin{itemize}
\item The overall number of steps needed, $s$; 
\item The block number;
\item The rounded number of steps done in this block (approximately $\frac{s}{\sqrt{n}}$ steps);
\item The total number of steps done until this block;
\item The starting configuration at this block;
\item The ending configuration at this block.
\end{itemize}
Furthermore, the number of variables needed is $2c$ plus a constant.

Boris can now challenge that some configuration is too long or that
the number of digits is wrong or that the information at the end of one
block does not coincide with the information at the start of another block
or that initial and final configurations are not starting and ending
configurations or select a block whose computation has to be checked
in the next round, again by distributing the steps covered in this
interval evenly onto $\sqrt n$ blocks in the next variable.

By $\bigO(c)$ iterations, the distance of steps between two neighbouring configurations becomes 1. Now Boris can select two pieces of information copied to check whether they are right or whether the \ls\ computation in the last step read the symbol correctly out of the input word and so on. These checks can all be done in $\log^*n + \bigO(1)$ steps.
\end{proof}

As yet, we have no characterisation of those problems in \np\ which are in 
$\baal[\bigO(\log^* n)]$ and we think that for each such problem it might depend heavily on 
the way the problem is formatted. The reason is that 
it may be difficult to even prove whether or not
$\p$ is contained in $\baal[\bigO(\log^*n)]$, due to the following proposition.
We will later show that the class $\paal[\log^* n + \bigO(1)]$ which is obtained from $\baal[\log^* n + \bigO(1)]$ by starting with one additional step which generates a variable of polynomially sized length coincides with \ph.

\begin{proposition}\label{prop:psubsetpspace}
If $\p \subseteq \baal[\bigO(\log^* n)]$, then $\p \subset \pspace$.
\end{proposition} 

\begin{proof}
By Theorem \ref{thm:baaltimedspace} (the condition $f(n) \geq n$ is not 
necessary for the first containment relation to hold), 
$\baal[\bigO(\log^*n)] \subseteq \dspace[\bigO((n+\log^*(n))\cdot \log^*(n))]$.
Moreover, by the space hierarchy theorem \cite[Corollary 9.4]{Sipser13}, one has 
$\dspace[\bigO((n+\log^*(n))\cdot \log^*(n))] \subset \dspace[\bigO(n^2)]$.
Thus if $\p \subseteq \linebreak[4]\baal[\bigO(\log^*n)]$, then 
$\p \subset \dspace[\bigO(n^2)] \subset \pspace$.
\end{proof}

\section{Polynomial-Size Padded Alternating Automatic Register Machine} 

An \aarm\ is constrained by the use of {\em bounded} automatic relations during each computation step.  As we will prove later in Theorem \ref{thm:ubaalelem}, this is a real limitation, for the use of unbounded automatic relations allows an \aarm\ to recognise the class of elementary recursive functions. 
In this section, we study the effect of 
allowing a polynomial-size padding to the input of an Alternating Automatic Register Machine on its time complexity; this new model of computation will be called a Polynomial-Size Padded Bounded Alternating Automatic Register Machine (\paarm).  The additional feature of a polynomial-size padding will sometimes be referred to informally as a ``booster'' step of the \paarm.  Intuitively, padding the input before the start of a computation allows a larger amount of information to be packed into the register's contents during a computation history.  We show two contrasting results: on the one hand, even a booster step does not allow an \paarm\ with time complexity $\bigO(1)$ to recognise non-regular languages; on the other hand, the class of languages recognised by \paarm s in time $\log^*n+\bigO(1)$ coincides with the polynomial hierarchy.   

Formally, a {\em Polynomial-Size Padded Bounded Alternating Automatic Register Machine} (\paarm) $M$ is represented as a quintuple $(\Gamma,\Sigma,A,B,p)$, where $\Gamma$ is the register alphabet, $\Sigma$ the input alphabet, $A$ and $B$ are two finite sets of instructions and $p$ is a polynomial.  As with an \aarm, the register $R$ initially contains an input string over $\Sigma$, and $R$'s contents may be changed in response to an instruction, $J \subseteq \Gamma^* 
\times \Gamma^*$ which is a bounded automatic relation.  
A computation history of a \paarm\ with input $w$ for any 
$w \in \Sigma^*$ is defined in the same way that was done for an \aarm, 
except that the initial configuration is       
$(\ell,wv,x)$ for some $I_{\ell} \in A$, some $x,v \in \Gamma^*$, 
$(wv,x) \in I_{\ell}$, where $v = @^k$
for a special symbol $@ \in \Gamma - \Sigma$ and $k \geq p(|w|)$.
Think of $@^k$ as padding of the input.
Anke's and Boris' strategies, denoted by $\cA$ and 
$\cB$ respectively, are defined as before.  For any $u \in \Gamma^*$, a winning strategy for Anke with respect to $(M,u)$ is also defined as before.  Given any $w \in \Sigma^*$, $M$ {\em accepts} $w$ if for every $v \in @^*$,
with $|v| \geq p(|w|)$, 
Anke has a winning strategy with respect to $(M,wv)$.
Similarly, $M$ {\em rejects} $w$ if for every $v \in @^*$,
with $|v| \geq p(|w|)$, 
Boris has a winning strategy with respect to $(M,wv)$.
Note that the winning strategies need to be there for 
every long enough padding. If Anke and Boris do not satisfy the above
properties, then $(A,B)$ is not a valid pair. 

\begin{definition}[Polynomial-Size Padded Bounded Alternating Automatic Register Machine Complexity]
\label{defn:paal}

Let $M = (\Gamma,\Sigma,A,B,p)$ be a \paarm\ and let $t \in \natnum_0$.  
For each $w \in \Sigma^*$, $M$ {\em accepts $w$ in time $t$} if for every $v = @^k$, where 
$k \geq p(|w|)$ and $@ \in \Gamma - \Sigma$, Anke has a winning strategy $\mathcal{A}$ with respect to $(M,wv)$ and for any strategy $\mathcal{B}$ played by Boris, the length of $\cH(\cA,\cB,M,wv)$ is not more than $t$.  

A \paarm\ decides a language $L$ in $f(n)$ steps for a function $f$ depending on the length $n$ of the input if for all $x \in \{0,1\}^n$, both players can enforce that the game terminates within $f(n)$ steps by playing optimally and one player has a winning strategy needing at most $f(n)$ moves and $x \in L$ if Anke is the player with the winning strategy. 
$\paal[f(n)]$ denotes the family of languages decided by \paarm s that decide in time $f(n)$.  


	\end{definition}
	
	\begin{remark*}
		Note that a $\paarm$-program can trivially simulate an $\aarm$-progr\-am by ignoring the generated padding; thus $\baal[\bigO(f(n))] \subseteq \paal[\bigO(\linebreak[4]f(n))]$. On the other hand, to simulate a booster step, an $\aarm$-program needs $\bigO(p(n))$ steps as each bounded automatic relation step can only increase the length by a constant.
	\end{remark*}

\noindent
As with $\baal[\bigO(f(n))]$, the class $\paal[\bigO(f(n))]$ is closed under the usual set-theoretic Boolean
operations as well as regular operations. 

\begin{theorem}\label{thm:closurepaal}
		Let $\star$ denote any one of the operations union, intersection and concatenation.
		For each $L_1 \in \paal[f(n)]$ and each $L_2 \in \paal[g(n)]$, $L_1 \star L_2 \in \paal[\max\{f(n),g(n)\}]$
		and $L_1^* \in \paal[f(n)]$.
		In particular, the languages in $\paal[f(n)]$ 
		are closed under union, intersection, concatenation 
	    and Kleene star. 
\end{theorem}

\noindent
One also has that any language recognised by a \paarm\ in constant time
must be regular; the converse assertion, that any regular language is recognised
by some \paarm\ in constant time, follows from an argument similar to
that for the case of \aarm s.  

	\begin{theorem}\label{the:regularpaal}
		For any constant $c$, all languages recognised by a $\paarm$ program in $c$ steps are regular.
	\end{theorem}

\def\proofofregularpaal{	
	\begin{proof}
		Let $L$ be any regular language. 
Then clearly, $\{x@^*: x \in L\}$ is also regular.
Now suppose that for some polynomial $p$, 
$B=\{x@^k: k\geq p(|x|), x \in A\}$ is regular.
Then, $A=\{x: (\exists j)(\forall k \geq j)[x@^k\in B]\}$.
As $B$ is regular, $A$ is regular as regular sets are closed under
first order quantification.
	\end{proof}
}

\proofofregularpaal

\begin{corollary}\label{cor:regularpaal}
$\paal[1] = \reg$.
\end{corollary}

\medskip
\noindent
Our main result concerning \paarm s is a characterisation of the polynomial hierarchy (\ph) as the class $\paal[\log^*n + \bigO(1)]$. 

\begin{theorem} \label{the:ph_paal_pspace}
		$\ph = \paal[\log^* n + \bigO(1)]$. 
	\end{theorem}

	\noindent
	To help with the proof, we first extend the Log-Star Lemma as follows.
Recall that a {\em configuration} (or {\em instantaneous description}) of a 
Turing Machine is represented by a string $xqw$, where $q$ is the current state of the machine and $x$ and $w$ are strings over the tape alphabet, such that the current tape contents is $xw$ and the current head location is the first symbol of $w$ \cite{HU79}. 
	
	\begin{lemma} \label{lem:turinglogstar}
	Checking the validity of a Turing Machine step, i.e., whether
a configuration of Turing Machine follows another
configuration (given as input, separated by a special separator symbol)
can be done in $\baal[\log^* n + \bigO(1)]$, where $n$ 
is the length of the shorter of the two configurations.
	\end{lemma}

\def\proofofturinglogstar{	
	\begin{proof}
		Let the input be the two configurations of the Turing Machine, 
where the second configuration is supposedly the successive step of the first one and separated by a separator symbol.  
Now there are two things that need to be checked: (1) The configuration is "copied" correctly from the previous step. Note that a valid Turing Machine transition will change only the cell on the tape head and/or both of its neighbour; thus "copied" here means the rest of the tape content should be the same; (2) The local Turing step is correct.
	
	\noindent
		For the first checking, the player who wants to verify, e.g. Anke, will give the infinite valid marking as used in the Log-Star Lemma. In addition, Anke also marks the position of the old tape head on the second configuration. Boris can then challenge the following:
(a) The Log-Star Lemma marking is not valid;
(b) The old tape head position is not marked correctly (in the intended position) on the second configuration;
(c) The string in a specified block differs, but not the symbols around the tape head;
(d) The length difference of the configuration is not bounded by a constant.
				
		\noindent
		Challenge (a) can be done in $\log^* n + \bigO(1)$ steps; this 
follows from the Log-Star Lemma. Challenge (b) can also be done in $\log^* n+\bigO(1)$ steps where both players reduce the block to focus on that position, and finally check whether it is on the same position or not. Challenge (c) can also be done in $\log^* n + \bigO(1)$ steps; this 
follows again from the Log-Star Lemma. Note that if Boris falsely challenges that the different symbol is around the tape head, Anke can counter-challenge by pointing out that at least one 
of its neighbours is a tape head. For challenge (d), note that a valid Turing Machine transition will only increase the length by at most one. Thus, Boris can pinpoint the last character of the shorter configuration and also its pair on another configuration, then check whether the longer one is only increased by up to one in length. This 
again can be done in $\log^* n + \bigO(1)$ steps.
		
		For the second checking about the correctness of the Turing step, it can be done in a constant number of checks as a finite automaton can check the computation and determine whether the Turing steps are locally correct, that is, each state is the successor state of the previous steps head position and the symbol to the left or right of the new head position is the symbol following from the transition to replace the old symbol and so on. Therefore, all-in-all the validity of a Turing Machine step can be 
checked in $\baal[\log^* n + \bigO(1)]$ steps.
	\end{proof}
}

\proofofturinglogstar
	
	\begin{proof}[Proof of Theorem \ref{the:ph_paal_pspace}]
	We first prove that $\paal[\log^* n + \bigO(1)] \subseteq \ph$.  Define a binary function $\tower$ recursively as follows:
	\[
	\begin{aligned}
	\tower(0,c) &= 1 \\
	\tower(d+1,c) &= 2^{c\cdot \tower(d,c)}.
	\end{aligned}
	\]
	We prove by induction that for each $c \geq 1$, there is a $c'$ such that for all $d$, $\tower(d+c',1) > \tower(d,c)$.  When $c = 1$, $\tower(d,c)$ gives the usual definition of the tower function.  In particular, when $c = 1$, one has $\tower(d+1,c) > \tower(d,c)$ for all $d$, so the induction statement holds for all $c = 1$ and all $d$.  Suppose that $c > 1$.  Then there is some $c'$ large enough so that $\tower(c',1) > c^2 = c^2 \cdot \tower(0,c)$, and so the induction statement holds for $d = 0$.  Assume by induction that $c > 1$ and that $\tower(c''+d,1) > c^2 \cdot \tower(d,c)$.  Then $2^{\tower(c''+d,1)} > 2^{c^2 \cdot \tower(d,c)} \geq (2^{c\cdot\tower(d,c)})^c > c^2 \cdot 2^{c\cdot\tower(d,c)}$, and therefore $\tower(c''+d+1,1) > c^2 \cdot \tower(d+1,c)$.  This completes the induction step.  

Say that Anke (resp.~Boris) {\em wins within $k$ steps} iff for some $n \leq k$, Boris (resp.~Anke) starts the $(n+1)$-st turn of the game and he (resp.~she) cannot make a move (it is assumed that both players aim to win and they play optimally).       	
Suppose that $c$ is the number of states of the automaton $M$ corresponding to the update function for the configuration of an \baal\ algorithm. Then there is a dfa of size at most $\tower(k+2,c)$ that recognises whether or not a player wins within $k$ steps. We prove this claim by induction on $k$, showing that there is a dfa of size $\tower(k+1,c)$ such that for any string $x$, the dfa accepts $x$ iff Anke wins within $k$ steps on input $x$ when she (resp.~Boris) starts the game. We can similarly show that there is a dfa of size $\tower(k+1,c)$ such that for any string $x$, the dfa accepts $x$ iff Boris wins within $k$ steps when he (resp.~Anke) starts, so the union of the languages accepted by the two dfa's is recognised by a dfa of size at most $\tower(k+2,c)$.  
      
By the earlier definition of ``winning within $k$ steps'', Anke cannot win in zero steps if she starts the game. If Boris starts, then Anke wins in zero steps on input $x$ iff there is no $y$ such that $M$ accepts $(x,y)$. The latter condition can be checked with a dfa of size $2^c = \tower(1,c)$: one first constructs from $M$ an nfa $M'$ that has the same set of $c$ states as $M$ and for any $x \in \Gamma$, $M'$ on input $x$ goes from state $p$ to state $q$ iff there is some $y \in \Gamma$ such that $M$ on input $conv(x,y)$ goes from state $p$ to state $q$; $M'$ is then converted into another nfa $M''$ with the same number of states as $M'$ that accepts the complement of the language accepted by $M'$, after which $M''$ can be converted into a dfa of size $2^c$.  This verifies the statement for the base case.

Assume inductively that for all $k' \leq k$, there is a dfa of size $\tower(k'+1,c)$ that accepts $x$ iff Anke wins within $k'$ steps on input $x$ when she (resp.~Boris) starts. Suppose it is Anke's turn and we need to check if she wins within $k+1$ steps. Let $M_k$ be a dfa of size $\tower(k+1,c)$ that accepts $x$ iff Anke wins within $k$ steps on input $x$ when Boris starts the game. Define an nfa $N$ as follows. For each state $p$ of $M$, make $\tower(k+1,c)$ states $(p,q_1),\ldots,(p,q_{\tower(k+1,c)})$, where $q_1,\ldots,q_{\tower(k+1,c)}$ are the states of $M_k$. Then each state $(p,q)$ on input $x \in \Gamma$ goes to each state $(p',q')$ such that for some $y \in \Gamma$, $M$ on input $conv(x,y)$ goes from state $p$ to state $p'$ and $M_k$ on input $y$ goes from state $q$ to state $q'$. The start state of $N$ is $(p_1,q_1)$, where $p_1$ and $q_1$ are the start states of $M$ and $M_k$ respectively, and the final states of $N$ are states $(p_f,q_f)$ such that $p_f$ and $q_f$ are final states of $M$ and $M_k$ respectively.  Then $N$ accepts $x$ iff there is a string $y$ such that $M$ accepts $(x,y)$ and Anke wins within $k$ steps on input $y$ when Boris starts; in other words, $N$ accepts $x$ iff Anke wins within $k+1$ steps when she starts with input $x$.   
	 The nfa $N$, which is of size $c \cdot \tower(k+1,c)$, can then be converted into a dfa $M'$ of size $2^{c \cdot \tower(k+1,c)} = \tower(k+2,c)$, as required. 
	
Suppose now that it is Boris' turn and we need to check if Anke wins within $k+1$ steps. By the induction hypothesis, there is a dfa $M'_k$ of size $\tower(k+1,c)$ that accepts $x$ iff Anke wins within $k$ steps on input $x$ when she starts. Let $M''_k$ be a dfa with the same number of states as $M'_k$ that accepts the complement of the language accepted by $M'_k$, that is, $M''_k$ accepts $x$ iff Anke does not win within $k$ steps on input $x$ when she starts. We build an nfa $N'$ with $c\cdot \tower(k+1,c)$ states following a construction similar to that of $N$ in the previous case but replacing $M_k$ by $M''_k$. Then $N'$ accepts $x$ iff there is a string $y$ such that $M$ accepts $conv(x,y)$ and Anke does not win within $k$ steps on input $y$ when she starts; in other words, $N'$ accepts $x$ iff Anke does not win within $k+1$ steps when Boris starts.  The nfa $N'$ can be converted into an nfa $N''$ with the same number of states as $N'$ that accepts the complement of the language accepted by $N'$ -- that is, $N''$ accepts $x$ iff Anke wins within $k+1$ steps when Boris starts.  The nfa $N''$ can in turn be converted into a dfa of size $2^{c\cdot \tower(k+1,c)} = \tower(k+2,c)$.  This completes the induction step.       
			
By the preceding result on the function $\tower$, $\tower(k+2,c)$ is bounded by $\tower(k+c',1)$ for some $c'$.  Thus any language in $\paal[\log^*n + \bigO(1)]$ is recognised in a constant number of alternating steps plus a predicate that can be computed by a dfa of size $\tower(\log^*n-3,1)$.  This dfa can be computed in \ls\ since $\log^*n$ can be computed in logarithmic space.  Then one constructs the dfa by determinizing out the last step until it reaches size $\log \log n$. This happens only when only constantly many steps are missing by the above tower result. These constantly many steps can be left as a formula with alternating quantifiers followed by a dfa computed in logarithmic space of size $\log \log n$. Thus the formula whether Anke wins is in \ph. Similarly for the formula whether Boris wins and so the overall decision procedure is in \ph.   
	
		For the proof that $\ph \subseteq \paal[\log^*n + \bigO(1)]$, we first note that \ph\ can be defined with alternating Turing machines \cite{Chandra76}. We define 
$\Sigma^P_k$ to be the class of languages recognised by alternating Turing Machine in polynomial time where the machine alternates between existential and universal states $k$ times starting with existential state. We also define $\Pi^P_k$ similarly but starting with universal state. \ph\ is then defined as the union of all $\Sigma^P_k$ and $\Pi^P_k$ for all $k \geq 0$. We now show $\Sigma^P_k \cup \Pi^P_k \subseteq \paal[\log^* n + \bigO(1)]$ for any fixed constant $k$. As the alternating Turing Machine runs in polynomial time on each alternation, 
the full computation (i.e., sequence of configurations)
in one single alternation can be captured 
non-deterministically in $p(m)$ Turing Machine steps, 
for some polynomial $p$ (which we assume to be bigger than linear), 
where $m$ is the length of the configuration at the start of the alternation.
In a $\paarm$-program, Anke first invokes a booster step to have a string of 
length at least $p^k(n)$. After that, Boris and Anke will alternately guess 
the 
full computation of the algorithm 
of length $p(p^i(n))$, $i=0,1,\ldots$, in their respective 
alternation: Boris guesses the first $p(n)$ computations 
(the first alternation), Anke then guesses the next $p(p(n))$ computations 
on top of it (the second alternation), etc. In addition, they also mark 
the position of the read head and symbol it looks upon in each step. Ideally, the $\paarm$-program will take $k$ alternating steps to complete the overall algorithm. 
Note that a $\paarm$ can keep multiple variables in the
register by using convolution, as long as the number of variables is a constant.
Thus, we could store the $k$ computations in $k$ variables: 
$v_1, v_2, \cdots, v_k$. Now each player can have the following 
choices of challenges to what the other player did:
(1) Copied some symbol wrongly from the input i.e. in $v_1$;
(2) Two successive Turing Machine steps in the computation are not valid (at some $v_i$);
(3) The last Turing Machine step on some computation (at some $v_i$) does not follow-up with the first Turing Machine step on the next computation (at $v_{i+1}$).
		
		\noindent
		All the above challenges can be done in
$\log^*n + \bigO(1)$ steps by a slight modification of Lemma \ref{lem:turinglogstar}. 
In particular, the third challenge needs one 
to compare the first Turing Machine configuration of $v_{i+1}$ and the 
last Turing Machine configuration $v_i$, which can be done in a way similar
to the proof of Lemma \ref{lem:turinglogstar}.
Thus, $\Sigma^P_k \cup \Pi^P_k \subseteq \paal[\log^* n + \bigO(1)]$ for every fixed 
constant $k$, therefore $PH \subseteq \paal[\log^* n + \bigO(1)]$. 
Corollary \ref{cor:aalpolyequalspspace} implies 
$\paal[\log^* n + \bigO(1)] \subseteq \pspace$.
	\end{proof}
	
\begin{remark*}
As $\paal[\log^*n + \bigO(1)] = \ph$, $\paal[\log^*n + \bigO(1)]$ is closed under Turing reducibility. Similarly one can show that $\paal[\bigO(\log^*n)]$ is closed under Turing reducibility.  After Anke invokes the booster step, Boris will guess the accepting computation together with all of the oracle answers. Anke then can challenge Boris on either the validity of the computation (without challenging the oracle) or challenge one of the oracle answers. Both challenges can be done in the same fashion as in Theorem \ref{the:ph_paal_pspace} but the latter needs one additional step to initiate the challenge of the oracle algorithm.  	
\end{remark*}
	
\noindent
In addition, we also get the following corollary about polynomial-time Turing reducibility. We recall that a {\em polynomial-time Turing reduction from problem $A$ to problem $B$} is an algorithm to solve $A$ in a polynomial number of steps 
by making a polynomial number of calls to an oracle solving $B$.
	
\begin{corollary}\label{cor:paalturingred}
		$\paal[\bigO(\log^* n)]$ is closed under polynomial-time Turing reduci\-bility.
	\end{corollary}
\def\proofofpaalturingred{	
	\begin{proof}
		 After Anke invokes the booster step, Boris will guess the accepting computation together with all of the oracle answers. Anke then can challenge Boris on either the validity of the computation (without challenging the oracle) or challenge one of the oracle answers. Both challenges can be done in the same fashion as in Theorem \ref{the:ph_paal_pspace} but the latter needs one additional step to initiate the challenge of the oracle algorithm.
	\end{proof}
}

\proofofpaalturingred

\noindent
In order to obtain the next corollary, we use the fact that the problem of deciding
$\tqbf_{f}$ -- the class of true quantified Boolean formulas with $\log^* n + f(n) + \bigO(1)$ 
alternations -- does not belong to any fixed level of the polynomial hierarchy (\ph) when 
\ph\ does not collapse. 

	\begin{corollary}\label{cor:paalsupsetph}
	If \ph\ does not collapse and $f$ is a logspace computable increasing and unbounded function, then $\baal[\log^*n + f(n) + \bigO(1)] \not\subseteq \ph$.
	\end{corollary}
\def\proofofpaalsupsetph{	
	\begin{proof}
	Suppose \ph\ does not collapse. Let $\tqbf_{f}$ be the problem $\{\phi: \phi~\mbox{is a}$ $\mbox{true fully quantified Boolean formula with $\log^*|\phi|+f(|\phi|)+\bigO(1)$ alternations}$ \linebreak[4] $\mbox{of quantifiers}\}$, and consider any Boolean formula.
	The computation with choice of variables is placed in $f(n)$ rounds of Anke and Boris moves followed by $\log^*n + \bigO(1)$ moves to evaluate the formula. Furthermore, if the formula has more than $f(n)$ alternations (and is thus invalid), any of the two players can challenge this in the first round and then it takes $\log^*n + \bigO(1)$ rounds to do the logspace computation of $f(n)$ and to verify that the formula does not belong to $\tqbf_f$.
	\end{proof}
} 

\proofofpaalsupsetph

\noindent
Finally, we observe that if $\ph = \pspace$, then (i) by Theorem \ref{the:ph_paal_pspace}, $\ph = \paal[\bigO(\log^*n)] = \pspace$; (ii) by Proposition \ref{prop:psubsetpspace}, $\p \not\subseteq \baal[\bigO(\log^*n)]$; thus $\baal[\bigO(\log^*n)]$ would be properly contained in $\paal[\bigO(\log^*n)]$.

\begin{proposition}\label{prop:aallogstarsubsetpaallogstar}
If $\ph = \pspace$, then $\baal[\bigO(\log^* n)] \subset \paal[\bigO(\log^* n)]$.
\end{proposition}

\begin{theorem}\label{thm:aalpaalreg}
If $f$ is monotonically increasing and unbounded, then $\baal[\log^* n\linebreak[3] - f(n)] = \reg$.
\end{theorem} 

\begin{proof}
Assume that there is an \aarm\ such that for each word $w$ there is either for Anke or for Boris a winning strategy of $\log^*n - f(n)$ steps. Then by the tower lemma, the resulting size of the dfa is $\bigO(\log\log n)$ for almost all $n$ and input words of length $n$. Thus the combined two dfas have at most size $poly(\log \log n)$ and there is no word $w$ on which not exactly one accepts in the given time. Now assume that for a given dfa of sufficient large $n$, there is a word $w$ where neither player succeeds in $\log^*n - f(n)$ rounds, where the $n$ is fixed. Due to the pumping lemma, on words of arbitary length with this property, one can pump down these words until they have size below $n$. However, such short words with this property do not exist by assumption.  Thus for this fixed $n$, all words of arbitrary length are accepted by computations of length $\log^*n - f(n)$. Thus the language is actually in $\baal[\bigO(1)]$ and in \reg.
\end{proof}

\section{Variants of \aarm s and \paarm s}

We conclude by discussing the computational power of
several variants of \linebreak[4]\aarm s and \paarm s.  The first variant,
which was alluded to earlier, is the use of {\em unbounded} automatic relations as update instructions.  More precisely, this means that
for the update automaton, there is no constant $c$ such that for any
given input parameter, the difference between the length of each possible output and that of the input parameter is at most $c$.      
In contrast to Theorem \ref{thm:baalre}, one can show that if an \aarm\ is allowed to use unbounded automatic relations as update instructions, then the class of languages recognised in time $\log^*n + \bigO(1)$ is 
precisely the class \ethree\ -- the class $\dtime(2^n) \cup \dtime\left(2^{2^n}\right) \cup \dtime\left(2^{2^{2^n}}\right) \cup \ldots$ of all sets computed by an elementary recursive function.   

\begin{theorem}\label{thm:ubaalelem}
Let $\ubaal[f(n)]$ denote the class of languages recognised in time
$f(n)$ by an \aarm\ that uses unbounded automatic relations as update instructions.  Then $\ethree = \ubaal[\log^*n + \bigO(1)]$. 
\end{theorem}

\begin{proof}
We first show that $\ethree \subseteq \ubaal[\log^*n + \bigO(1)]$.
The main idea of the proof is based on Lemma \ref{lem:turinglogstar}.
As shown in the lemma, an \aarm\ can verify the validity of a transition between two Turing machine configurations in $\log^* n + \bigO(1)$ steps, where $n$ is the length of the shorter configuration.
Suppose that the running time of a Turing machine $M$ is at most $2^n$ for any given input of length $n$. The set of all pairs $(x,y)$ such that $y$ represents a (possibly partial) computation history of $M$ is an automatic relation, and so it can be used to update the (unbounded) \aarm\ register. For any given input, Anke guesses the full computation history of the Turing machine; this can be represented by a string of length $\bigO(2^n)$.  An argument similar to that in the proof of Lemma \ref{lem:turinglogstar} then shows that an \aarm\ can verify in $\log^*\left(\bigO(2^n)\right) + \bigO(1) \leq \log^*\left(2^{2^n}\right) + \bigO(1) = \log^*n + 2 + \bigO(1) = \log^*n + \bigO(1)$ steps whether or not the guessed computation history is valid.  A similar argument applies if the running time of the Turing machine is $2^{2^{2^{\ldots}}}$ for any given number of exponentiations.  One can clock the steps done with a computation of $\log^* n$ for input size $n$ (taking $\log^*n$ steps) and then run a constant number of additional steps, ending with either Anke winning or Boris winning if the timebound is not kept.

Next, we show that $\ubaal[\log^*n + \bigO(1)] \subseteq \ethree$. 
Given any language recognised by an unbounded \aarm\ with time complexity 
$\log^*n + \bigO(1)$, as shown in the proof of Theorem \ref{the:ph_paal_pspace}, if $c$ is the number of states of the update automaton, then a dfa of size at most $\tower(\log^*n + \bigO(1),c)$ recognises whether or not a given player wins within $\log^*n + \bigO(1)$ steps.  This dfa can be computed in 
$\tower(\log^*n + \bigO(1),c)$ time, and so the given language belongs to
\ethree.  
\end{proof}

\begin{remark}
Suppose that $f$ is \ls\ computable, increasing and unbounded;
an example of such a function is $\log(\log^* n)$.
Then $\baal[\log^*n + f(n) + \bigO(1)]$ contains the class
$\nspace[\log n \cdot f(n)]$ and $\uaal[\log^*n + f(n) + 1]$
contains a language not in \ethree. The first result shows
that under the assumption $\ph = \ls$, the class
$\baal[\log^*n + f(n) + \bigO(1)]$ contains a language outside
$\ph$ and but still in $\pspace$.
\end{remark}

\medskip
\noindent
One might be interested in $\paarm$-programs for which the number of alternations 
is constant but a player can do more than one automatic function 
step in a single turn. In particular, we can show that if, after the booster step the number of alternation is just two, then every language 
in $\nll$, which is closed under complementation, 
can be recognised in $\bigO(\log \log n)$ steps. After padding the input,
Boris will run the nondeterministic logspace computation where he 
uses location number for the input cell read in a step. Boris can
code this computation in the padded space, where each configuration step
is of logarithmic size.
Now Anke verifies the computation, which can be done in $\bigO(\log \log n)$ 
deterministic steps. The idea is to 
guess the two symbols from two successive logspace computation steps which 
contradict 
a Turing machine computation (that is, there are errors in the Turing machine computation). After that, Anke verifies that the positions of the error symbols are indeed correct by marking off 
every other 
symbol in each logspace sized protocol. The marking off
of steps 
is repeated until all symbols except the errors in the logspace sized protocols are marked off, while maintaining 
the same parities of the errors' positions. 
It is easy to see that Anke needs $\bigO(\log \log n)$ marking off steps in total, and one more 
step to confirm the error.

\medskip
\noindent
Another possible variant of an \aarm\ uses transducers as basic operations instead of bounded automatic relations. We define an Alternating Transducer Register Machine (ATRM) and a Polynomial-Size Padded Alternating Transducer Register Machine (PATRM) similarly to an $\aarm$ and a $\paarm$ respectively. When using bounded automatic relations, the 
``bottleneck'' of the computational power, as shown in this paper, lies in the string comparison, which needs 
$\bigO(\log^* n)$ steps by the Log-Star Lemma and its extension. This, however, does not apply when transducers are used as the string comparison can be done in constantly many transducer steps. Thus, we can show that not only an NP-complete problem is recognised by ATRM in constantly many steps but also PH is contained in PATRM in constantly many steps.
In fact, denoting the class of languages recognised by an ATRM in time
$t(n)$ by $\uatl[t(n)]$, one can show the following.

\begin{theorem}\label{thm:uatlah}
$\uatl[\bigO(1)]$ is the union of the classes in the arithmetical hierarchy.
\end{theorem}

\begin{proof}
We first note that the class R.E.\ of recursively enumerable languages
coincides with the class of languages obtained by applying constantly many
transducer steps with existential guessing \cite[Lemma 3.1]{Mateescu95}.
The alternating closure then gives that $\uatl[\bigO(1)]$ is equal to the
whole class of languages in the arithmetical hierarchy.  
\end{proof}

\medskip
\noindent 
Note that the arithmetical hierarchy consists of undecidable sets (except for the lowest level). Thus,
owing to Theorem \ref{thm:uatlah}, we believe that transducers 
are too powerful to be used as basic operations and 
that the use of bounded automatic relations as basic operations 
gives a more 
appropriate model of computation. 

\medskip
\noindent
The main results on complexity classes defined by \aarm s are
summed up in Figure \ref{fig:summary1} while those on complexity 
classes defined by \paarm s are summed up in Figure \ref{fig:summary2}.
For any function $f$, 
$\baal[f(n)]$ denotes the class of languages recognised by an 
\aarm\ in $f(n)$ time. 
Arrows are labelled with references to the corresponding results or definitions; well-known inclusions can be found in standard textbooks \cite{HU79,Sipser13} or the Complexity Zoo (\url{https://complexityzoo.net/Complexity_Zoo}).
Table \ref{table:summary} also summarises the main results. 

\vspace*{-.2cm}

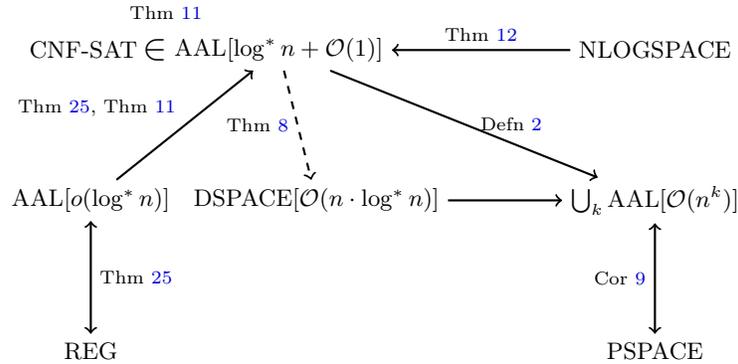
\begin{figure}[H]
\begin{center} \setlength{\unitlength}{1in}
\begin{tikzpicture}

\node at (-4.5,1) (baalologstarn) {$\baal[o(\log^* n)]$};
\node at (-4.5,-1) (reg) {\reg};
\node at (-2,3) (baallogstar) {$\baal[\log^* n+\bigO(1)]$};
\node at (-4.6,3) (sat) {\cnfsat};
\node at (-3.65,3) (satinaal) {$\Scale[1.5]{\in}$};
\node at (-3.5,3.5) (satinaalref) {\scriptsize{Thm \ref{the:sat_baal}}};
\node at (3,3) (nlogspace) {$\nll$};
\node at (3,1) (baalpoly) {$\bigcup_k \baal[\bigO(n^k)]$};
\node at (3,-1) (pspace) {\pspace};
\node at (-1.5,1) (dspacelogstarntimesn) {$\dspace[\bigO(n\cdot \log^* n)]$};


\draw[->,thick]
(baalologstarn) -- (baallogstar) node[midway,above left] {\scriptsize{Thm \ref{thm:aalpaalreg}, Thm \ref{the:sat_baal}}};

\draw[->,thick,dashed]
(baallogstar) -- (dspacelogstarntimesn) node[midway,left] {\scriptsize{Thm
\ref{thm:baaltimedspace}}};

\draw[->,thick]
(nlogspace) -- (baallogstar) node[midway,above] {\scriptsize{Thm
\ref{thm:nlsinaal}}}; 

\draw[->,thick]
(dspacelogstarntimesn) -- (baalpoly);

\draw[<->,thick]
(baalologstarn) -- (reg) node[midway,right] {\scriptsize{Thm \ref{thm:aalpaalreg}}};

\draw[<->,thick]
(baalpoly) -- (pspace) node[midway,left] {\scriptsize{Cor
\ref{cor:aalpolyequalspspace}}};

\draw[->,thick]
(baallogstar) -- (baalpoly) node[midway,right=1mm] {\scriptsize{Defn \ref{defn:aal}}};




\end{tikzpicture}
\end{center}
\caption{Relationships between complexity classes/\cnfsat. A solid arrow from $X$ to $Y$ 
means that $X$ is a proper subset of $Y$.  A double-headed solid arrow between $X$
and $Y$ means that $X$ is equal to $Y$.  If $X$ is a subset
of $Y$ but it is not known whether they are equal sets,
then the arrow is dashed.}
\label{fig:summary1}
\end{figure}   

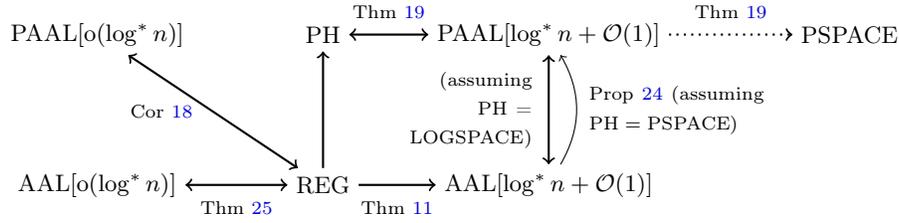
\begin{figure}[H] 
\begin{center} \setlength{\unitlength}{1in} 
\begin{tikzpicture}

\node at (-3,0) (ph) {\ph};
\node at (0,-2) (baallogstar) {$\baal[\log^* n + \bigO(1)]$};
\node at (0,0) (paallogstar) {$\paal[\log^* n + \bigO(1)]$};
\node at (-6,-2) (baalologstar) {$\baal[\smallO(\log^* n)]$};
\node at (-6,0) (paalologstar) {$\paal[\smallO(\log^* n)]$};
\node at (4,0) (pspace) {\pspace};
\node at (-3,-2) (reg) {\reg};

\draw[<->,thick]
(baallogstar) -- (paallogstar) node[midway,left=1mm, align=right, text width=1.8cm] {\scriptsize{(assuming \ph\ = \ls)}};

\path[->]  (baallogstar) edge [bend right]  node [midway,right=.5mm, text width=3cm] {\scriptsize{Prop
\ref{prop:aallogstarsubsetpaallogstar} (assuming $\ph = \pspace$)}} (paallogstar);

\draw[<->,thick]
(ph) -- (paallogstar) node[midway,above=1mm] {\scriptsize{Thm
\ref{the:ph_paal_pspace}}};

\draw[->,thick,dotted]
(paallogstar) -- (pspace) node[midway,above=1mm] {\scriptsize{Thm
\ref{the:ph_paal_pspace}}};

\draw[->,thick]
(reg) -- (baallogstar) node[midway,below=1mm] {\scriptsize{Thm
\ref{the:sat_baal}}};

\draw[<->,thick]
(paalologstar) -- (reg) node[midway,left=1.5mm] {\scriptsize{Cor \ref{cor:regularpaal}}};

\draw[<->,thick]
(baalologstar) -- (reg) node[midway,below=1mm] {\scriptsize{Thm \ref{thm:aalpaalreg}}};

\draw[->,thick]
(reg) -- (ph) node[midway,right=1mm]{~}; 

\end{tikzpicture}
\end{center}
\caption{Relationships between complexity classes.} 
\label{fig:summary2}
\end{figure}   

\begin{center}
\begin{table}[H]
\begin{tabular}{|c|c|c|c|}
\hline
\multirow{2}{*}{Step Numbers} & \multicolumn{3}{c|}{Criterion} \\ \cline{2-4}
\multirow{1}{*}{\makebox[2.8cm]{}} & \makebox[2.8cm]{\makecell{\baal}} & \makebox[2.8cm]{\makecell{\paal}} &\makebox[2.8cm]{\makecell{\uaal}}  \\
\hline
\multirow{1}{*}{\makebox[2.8cm]{\makecell{$\bigO(1),\,\smallO(\log^*n)$}}}
&\makebox[2.8cm]{\makecell{$=\reg$}}&\makebox[2.8cm]{\makecell{$=\reg$}}&\makebox[2.8cm]{\makecell{$=\reg$}}
\\ \hline
\multirow{1}{*}{\makebox[2.8cm]{\makecell{$\log^*n + \bigO(1)$}}}
&\makebox[2.8cm]{\makecell{$\supseteq \ls$}}&\makebox[2.8cm]{\makecell{$=\ph$}}&\makebox[2.8cm]{\makecell{$=\ethree$}}
\\ \hline
\multirow{3}{*}{\makebox[2.8cm]{\makecell{$poly(n)$ \\ ~ \\ ~}}}
&\makebox[2.8cm]{\makecell{$=\pspace$}}&\makebox[2.8cm]{\makecell{$=\pspace$}}&\makebox[2.8cm]{\makecell{\ethree \\ $\subset \uaal \subset$ \\ \efour}}
\\ \hline
\end{tabular}
\caption{Summary of results. \ethree\ denotes the class of all sets computed by an elementary function and \efour\ denotes the class of all sets computed by functions on level 4 of the Grzegorczyk Hierarchy.}\label{table:summary}
\end{table}
\end{center}

\end{document}